\documentclass[twocolumn,journal]{IEEEtran}
\usepackage{verbatim}
\usepackage{amsmath}
\usepackage{amsthm}
\usepackage{amssymb}
\usepackage{color}
\usepackage{bm}
\usepackage{multirow}

\usepackage[pdftex]{graphicx}

\newtheorem{theorem}{Theorem}
\newtheorem{corollary}{Corollary}
\newtheorem{remark}{Remark}
\newtheorem{lemma}{Lemma}

\def\QED{\mbox{\rule[0pt]{1.5ex}{1.5ex}}}

\renewcommand{\qed}{\hfill \QED}

 \newenvironment{proofof}[1]{\vspace*{5mm} \par \noindent
         \quad{\it Proof of #1:\hspace{2mm}}}{\qed
}

\def\FF{\mathbb{F}}
\def\ZZ{\mathbb{Z}}

\def\Label#1{\label{#1}\ [\ \text{#1}\ ]\ }
\def\Label{\label}

\makeatletter
\def\mojiparline#1{
    \newcounter{mpl}
    \setcounter{mpl}{#1}
    \@tempdima=\linewidth
    \advance\@tempdima by-\value{mpl}zw
    \addtocounter{mpl}{-1}
    \divide\@tempdima by \value{mpl}
    \advance\kanjiskip by\@tempdima
    \advance\parindent by\@tempdima
}
\makeatother

\begin{document}

\title{Secure network code over one-hop relay network}

\author{Masahito Hayashi \IEEEmembership{Fellow, IEEE}
and Ning Cai \IEEEmembership{Fellow, IEEE}
\thanks{The work of M. Hayashi was supported in part by
the Japan Society of the Promotion of Science (JSPS) Grant-in-Aid 
for Scientific Research (B) Grant 16KT0017 and
for Scientific Research (A) Grant 17H01280 and
for Scientific Research (C) Grant 16K00014,
in part by the Okawa Research Grant, and 
in part by the Kayamori Foundation of Informational Science Advancement.
The material in this paper was presented in part at the 2017 IEEE International Symposium on Information Theory (ISIT 2017),   Aachen (Germany), 25-30 June 2017 \cite{HOKC}.}
\thanks{Masahito Hayashi is with the Graduate School of Mathematics, Nagoya University, Nagoya, 464-8602, Japan. 
He is also with 
Shenzhen Institute for Quantum Science and Engineering, Southern University of Science and Technology,
Shenzhen, 518055, China,
Center for Quantum Computing, Peng Cheng Laboratory, Shenzhen 518000, China,
and the Centre for Quantum Technologies, National University of Singapore, 3 Science Drive 2, 117542, Singapore
(e-mail:masahito@math.nagoya-u.ac.jp).
Ning Cai is with the School of Information Science and Technology, ShanghaiTech University, Middle Huaxia Road no 393,
Pudong, Shanghai  201210, China
(e-mail: ningcai@shanghaitech.edu.cn).} }

\markboth{M. Hayashi 
and N. Cai: Secure network code over one-hop relay network}{}

\maketitle

\begin{abstract}
When there exists a malicious attacker in the network,
we need to consider the possibilities of eavesdropping and the contamination simultaneously.
Under an acyclic broadcast network, the optimality of linear codes was shown 
when Eve is allowed to attack any $r$ edges.
The optimality of linear codes is not shown under a different assumption for Eve.
As a typical example of an acyclic unicast network,
we focus on the one-hop relay network under the single transmission scheme
by assuming that Eve attacks only one edge in each level.
Surprisingly, as a result, we find that 
a non-linear code significantly improves the performance on the one-hop relay network
over linear codes.
That is, a non-liner code realizes the imperfect security on this model that cannot be realized by linear codes.
This kind of superiority of a linear code still holds even with considering the effect of sequential error injection on information leakage.
\end{abstract}


\section{Introduction}
Security of information transmission over a network is a crucial problem
because there is a risk that 
a part of channels on the network are attacked by a malicious third party.
Secure network coding enables us to guarantee the security even under the existence of such an attacker.
Cai and Yeung \cite{Cai2002,CY,YN} discussed the secrecy for the malicious
adversary, Eve, wiretapping a subset $E_E$ of all channels in the network. 
The papers \cite{CG,RSS,FMSS,NYZ,HY,Cai,CC,AVF16-1,AVF16-2,AVF17} developed several types of secure network coding.
Combining the codes in \cite{Cai2002} and \cite{CY2},
the paper \cite{NY} proposed a linear code to keep the secrecy of the message 
and the robustness from the injection of error (contamination) simultaneously.
Like traditional error-correcting code and error correction network code (i.e. against Byzantine attack) in \cite{YC2,CY2}, 
the paper \cite{NY} considered 
the robustness in the worst case, or equivalently 
it evaluated the error probability when the adversary to inject error knows the message to be sent. 
The papers \cite{Yao2014,SK,KMU} also made similar studies.
Also, the papers \cite{SK,KMU,Matsumoto2011,Matsumoto2011a} showed the existence of a secrecy code that
universally works for any type of eavesdroppers under the size
constraint of $E_E$. 

\begin{figure}[htbp]
\begin{center}
\includegraphics[scale=0.5]{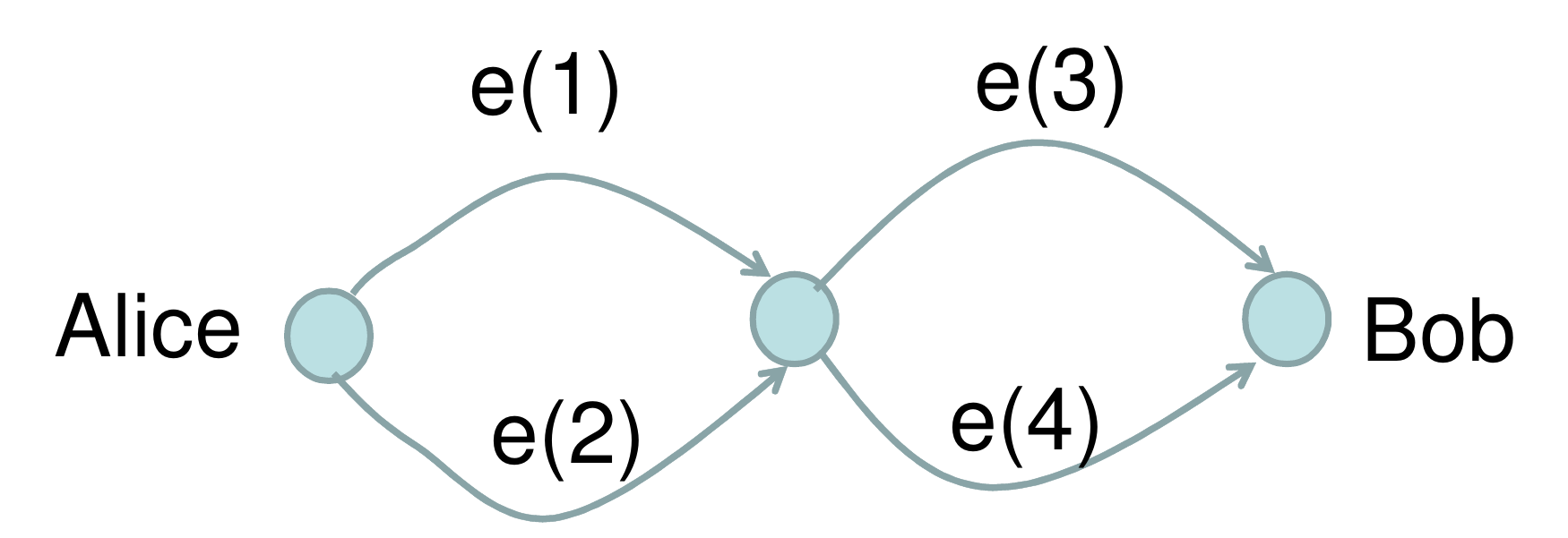}
\end{center}
\caption{one-hop relay network.}
\Label{F1}
\end{figure}%

However, two issues were not sufficiently studied in existing papers.
The first issue is the optimality of linear code.
Wile this optimality was shown in the papers \cite{LYC,CY,CY2}
under a general acyclic network with a single source and multiple receivers
(acyclic broadcast network),
they allowed Eve to attack any $r$ edges in the network (See Table \ref{TCom}).
But in many cases, linear codes are not optimal, or in other words, non-linear code has better performance, for example, in coding for multiple source network and classical error-correcting code (which can be considered as error correction network coding for a two-node network, 
``point-to-point network)''. 
In several known examples of multiple source network, non-linear code can do better than linear one \cite{DFZ}.
Only a limited number of studies \cite{DFZ,CG,YYZ} discussed non-linear codes over networks.
Hence, it has been an open problem whether 
the optimality of linear code still holds for an acyclic broadcast network even under a different assumption for Eve.
That is, if we assume a different type of assumption is imposed for Eve,
there is a possibility that a non-linear code overcomes linear codes.
To resolve this problem,
in this paper, as a typical example for unicast networks,
we focus on the one-hop relay network (Fig. \ref{F1}) 
when the sender sends a single element of a prime field $\FF_p$ 
(the single transmission scheme).
In this network, it is natural to assume that Eve can attack only one edge in each level
because any code is insecure for attacking both edges in one level. 
Hence, assuming that Eve attacks one edge before and after the intermediate node,
we study whether a non-linear code overcomes linear code in this model by adapting the imperfect secure criterion.

In fact, the perfect security is too restrictive for the single transmission scheme.
It is known that there exists an imperfectly secure linear code 
over a finite field $\FF_q$ of sufficiently large prime power $q$
when Eve may access a subset of channels that does not contain a cut between Alice and Bob 
even when the linear code does not employ private randomness in the intermediate nodes \cite{Bhattad}.
Hence, we adopt the imperfectly secure criterion.

 \begin{table*}[!t]
\caption{Summary of comparison with existing results for single transmission scheme of a finite field}
\label{TCom}
\begin{center}
  \begin{tabular}{|l|c|c|c|c|} 
\hline
& type of & type of &  \multirow{2}{*}{Field}   & optimality of\\
& network & attack  &  & linear code\\
\hline
\multirow{2}{*}{Paper \cite{LYC}} 
&  acyclic broadcast    & \multirow{2}{*}{no attack} &  sufficiently  & \multirow{2}{*}{Yes} \\
&  network  &  &  large &  \\
\hline
\multirow{2}{*}{Paper \cite{CY}} 
&  acyclic broadcast    & any $r$ edges &  sufficiently  & \multirow{2}{*}{Yes} \\
&  network  & are eavesdropped &  large &  \\
\hline
\multirow{2}{*}{Paper \cite{CY2}} 
&  acyclic broadcast    & any $r$ edges &  sufficiently  & \multirow{2}{*}{Yes} \\
&  network  & are injected errors &  large &  \\
\hline
\multirow{2}{*}{Examples in \cite{DFZ}} 
&  acyclic multicast  & \multirow{2}{*}{no attack}  &  sufficiently  & \multirow{2}{*}{No} \\
&  network  &  &  large &  \\
\hline
\hline
\multirow{2}{*}{passive attack}& one-hop & 1 edge in each level  & 
\multirow{2}{*}{any finite field}  & \multirow{2}{*}{No} \\
 & relay network & is eavesdropped  &   &  \\
\hline
active attack& one-hop & 1 edge in each level  & prime field  & no secure \\
with $\FF_2$ & relay network & is injected errors  & $\FF_2$  &  code\\
\hline
active attack& one-hop & 1 edge in each level  & finite field $\FF_q$  & \multirow{2}{*}{No} \\
(Characteristic is not $2$)
 & relay network & is injected errors  & ($q \neq 2^l$)  &  \\
\hline
  \end{tabular}
\\
 \end{center}
\vspace{2ex}
Since this table addresses the single transmission scheme on a finite field,
the linearity means the scalar linearity.
The case with the vector linearity was discussed in another recent paper \cite{CH},
which discussed multilayer networks including one-hop relay network.
\end{table*}

The second issue is the effect of Eve's contamination to help to have more information about the message. 
Indeed the existing studies \cite{NY,Yao2014,SK,KMU} evaluated errors 
only when the information on a part of the network is changed,
but they evaluated the secrecy only when the information on the network is not changed\footnote{In other words, 
it seemed to be considered at that time that injecting error might only 
make Eve's wiretapping the message more difficult.} 
or Eve did not know the replaced information.
On the other hand, as another model, we assume that the goal of Eve to inject error is to help to have more information about the message. 
In this case, she may inject error according to the knowledge which she obtained from her previous action but not the message. 
This improvement for her ability of eavesdropping is essential,
and this type of attack is called an {\it active attack}
while the attack without this improvement is called a {\it passive attack}.
The recent paper \cite{HOKC1} discussed this model, i.e., evaluated the secrecy as well as the error when Eve contaminates the eavesdropped information and knows the replaced information. 
Then, it showed that 
Eve's noise injection cannot improve her performance of eavesdropping
for linear network codes.
This kind of discussion over linear network codes was extended to a multilayer network \cite{CH}.
The effects of Eve's contamination depend on the type of network code.
This kind of reduction of security analysis was not shown for non-linear codes.
Therefore, we study the superiority of non-linear codes under both settings, i.e., passive attack and active attack.


First, we show the non-existence of imperfectly secure linear code over this network.
Then, in the binary case, we propose a non-linear code on this network, in which
even if Eve makes a passive attack, she cannot recover the original message, i.e., 
the imperfect security holds.
However, when she makes an active attack before the intermediate node,
she can recover the original message.
Similar unexpected properties for a nonlinear network error-correcting code were reported in \cite{YYZ}.
Also,
we show that the network code is limited to this example
when we impose several natural secrecy conditions on the code over this network in the binary case.
This discussion shows that no code can guarantee the security over this type of active attack 
on the one-hop relay network in the binary case.
However, in the ternary case, there exists a code such that Eve cannot completely recover the message even with this type of active attack on the one-hop relay network.
To discuss this problem, we introduce a new concept an ``anti-Latin square'', which is an opposite concept to a Latin square.
Therefore, the goal of this paper is summarized as the derivation of the following three facts over the one-hop relay network.
(1) No linear code is imperfectly secure. 
(2) In the binary case, there exists an imperfectly secure non-linear code for passive attacks, 
but no non-linear code is imperfectly secure for active attack
(3) In the ternary case, there exists an imperfectly secure non-linear code even for active attacks.
The relation with existing studies are summarized in Table \ref{TCom}.

Indeed, the assumption for the single transmission scheme is crucial.
The linearity of this setting is called the scalar linearity
In contrast, the linear setting is called the vector linearity when the sender sends multiple elements of a finite field\footnote{In fact, the paper \cite{Bhattad} discussed this kind of imperfectly secure code under scalar linearity 
with a different type of network while it chooses large $q$. 
In contrast, the paper \cite{SK} discussed a similar imperfectly secure code construction
by increasing the number of transmitted elements of the finite field (the vector linearity) 
while it did not increase the size of $q$.
The paper \cite{CCMG} extended this type of vector linearity setting 
of imperfectly secure codes to the case with multi-source multicast.}.
Our recent paper \cite{CH} discussed multilayer networks including the one-hop relay network.
However, it addressed the case when multiple letters are transmitted, in which, the linearity means vector linearity.
Then, the optimality of vector linearity was shown in such a setting.

Here, we remark our setting for the intermediate node.
In our setting, only the source node is allowed to employ scramble random numbers
and the intermediate node is not allowed to use it
due to the following reason\footnote{
Since stochastic encoder is essential for secure information transmission, 
we need to be careful for the availability of private randomness in each node.
However, a stochastic decoder cannot improve the performance even for 
secure information transmission due to the following reason.
The decoding error probability preserves the convex combination for the decoder.
Hence, the decoding error probability of 
a stochastic decoder can be written as a convex combination of 
the decoding error probabilities of 
deterministic decoders.}.
A physical device is required for 
the realization of physical random numbers as scramble random numbers.
If a high quality physical random number is needed,
the physical device is expensive and/or 
consumes a non-negligible space because
it often needs high-level quantum information technologies with advanced security analysis \cite{HG,HZ}.
It might be possible to prepare such devices on the source side.
However, it increases the cost to prepare devices in the intermediate nodes
because networks with such devices require more complicated maintenance than a conventional network.
Therefore, from the economical reason, it is natural to impose this constraint to our network code
while only a few papers \cite{CY-07,CHK-10,CHK-13,CH} discussed this type of restriction.

The remaining of this paper is organized as follows.
Section \ref{S4-1} gives our formulation over the one-hop relay network.
Next, Section \ref{S4-2} shows the non-existence of imperfectly secure linear code over this network.
Then, Section \ref{S3} discusses the binary case.
Finally, Section \ref{S3-5} studies the case when message size is a general number $d$,
which includes the ternary case.

\section{Formulation}\Label{S4-1}
In this paper, we focus on the imperfect security, i.e.,
the property that Eve cannot recover the original message with probability one \cite{CN11,CN11b}.
To cover a large possibility, we consider the two cases.
In one case, the message ${\cal X}$ set is a finite field $\FF_q$.
In the other case, it is a quotient ring $\ZZ_d$. 
The case  of a quotient ring $\ZZ_d$ contains the case of a prime field $\FF_p$,
which is a special case of a finite field$\FF_q$.

Then, we consider the secure codes over the one-hop relay network given in Fig. \ref{F1} with 
edges $E=\{e(1),e(2),e(3),e(4)\}$
only in the single transmission scheme, i.e., the case when
the sender sends only one element of ${\cal X}$.
The linearity under this setting is called the scalar linearity 
when ${\cal X}$ is a field $\FF_q$ \cite[Section I]{DFZ}.
That is, we consider the transmission of the message $M$ in ${\cal X}$ by using the one-hop relay network given in Fig. \ref{F1}
when the information at the edges is given as an element of ${\cal X}$.
Here, $d$ is an arbitrary natural number, and it is not necessarily a prime number.
Also, private randomness is allowed in the sender, but
no private randomness is allowed in the intermediate nodes.
Therefore, 
we are allowed to choose an arbitrary deterministic function
$\varphi$ from ${\cal X}^2$ to ${\cal X}^2$
as our coding operation on the intermediate node.
Our encoder in the source node is given as a stochastic map $\phi$
from ${\cal X}$ to ${\cal X}^2$,
and our decoder is given as a deterministic function
$\psi$ from ${\cal X}^2$ to ${\cal X}$.
Eve is allowed to attack the subset $E_E$ of two edges of $E$ except for the pairs $\{e(1),e(2)\}$ and $\{e(3),e(4)\}$.
In this section, 
we call the triplet $(\phi,\varphi,\psi)$ a code over the one-hop relay network (Fig. \ref{F1}).

We have two attack models, the {\it passive attack} and the {\it active attack}.
In the passive attack, Eve can eavesdrop two edges, but cannot change the information on the attacked edge.
In the active attack, 
Eve can insert another information on the attacked edge in the first group $\{e(1),e(2)\}$,
and eavesdrop one edge in the second group $\{e(3),e(4)\}$.
Here, Eve cannot change the edge to be attacked by using the information 
on the attacked edge in the first group $\{e(1),e(2)\}$.
When a code satisfies the following two conditions in the respective models,
the code is called imperfectly secure in the respective models.
Otherwise, it is called insecure in the respective models.
In the following conditions, the information on the edge $e(i)$ is written as $Y_i$.

\begin{description}
\item[(B1)] (Recoverability) 
Bob can recover the message $M$ 
from $Y_3 $ and $Y_4$ when Eve does not make any replacement.

\item[(B2)] (Secrecy)
No active attack $\tilde{\psi}$ from ${\cal X}^3$ to ${\cal X}$ 
satisfies one of the following conditions.
\begin{align}
\tilde{\psi}(Y_1,Y_1',Y_3)=M,\quad
\tilde{\psi}(Y_1,Y_1',Y_4)=M, \\
\tilde{\psi}(Y_2,Y_2',Y_3)=M,\quad
\tilde{\psi}(Y_2,Y_2',Y_4)=M,
\end{align}
where $Y_1'$ and $Y_2'$ are the information replaced by Eve at the edges $e(1)$ and $e(2)$,
and $Y_3$ and $Y_4$ are the information at the edges $e(3)$ and $e(4)$.
This kind of secrecy is called imperfect security \cite{CN11,CN11b}.
\end{description}

\section{Linear code}\Label{S4-2}
First, we address linear codes over the one-hop relay network
only when the message set ${\cal X}$ is a finite field $\FF_q$.
As shown in the recent paper \cite{HOKC1},
Eve's noise injection cannot improve her performance of eavesdropping
for linear network codes.
Hence, it is sufficient for analysis of information leakage
to discuss the case when Eve does not contaminate the eavesdropped information.
Since the message is an element of $\FF_q$,
the linearity in this problem can be regarded as scalar linearity \cite[Section I]{DFZ}.
The following theorem shows the impossibility of secure communication by using linear codes.

\begin{theorem}\Label{TT5}
There is no secure linear code even for passive attack 
when the message set ${\cal X}$ is a finite field $\FF_q$.
\end{theorem}
\begin{proof}
Due to the linearity of our code, 
we can choose 
a $2 \times 2$ matrix $A$ on $\FF_p$ such that
the relations $ Y_3= A_{1,1}Y_1 +A_{1,2}Y_2$
and
$ Y_4= A_{2,1}Y_1 +A_{2,2}Y_2$
holds when there is no attack.

First, we assume that $(A_{1,1},A_{1,2})$ is $(0,0)$.
Eve can obtain all the information Bob obtains when Eve eavesdrops $e(4)$.
Since Bob can recover the message, Eve can also do it.

Second, we assume that the vector $(A_{1,1},A_{1,2})\neq (0,0)$ is a constant times of $(1,0)$.
When Eve eavesdrops $e(2)$ and $e(3)$, 
Eve can obtain the information of $Y_1$ and $Y_2$, which contains the information that Bob obtains.
Since Bob can recover the message, Eve can also do it.
We obtain the same conclusion when $(A_{1,1},A_{1,2})\neq (0,0)$ is 
a constant times of $(0,1)$ and Eve eavesdrops $e(1)$ and $e(3)$.

Third, we assume that the vector $(A_{1,1},A_{1,2})\neq (0,0)$ is not a constant times 
of $(1,0)$ nor $(1,0)$.
When Eve eavesdrops $e(3)$ and $e(1)$,
Eve can recover the information $Y_4$ from $Y_3$ and $Y_1$.
Hence, Eve obtains all information that the receiver gets.
\end{proof}

\section{Analysis with $\FF_2$}\Label{S3}

The aim of this section is to prove the following theorem.
\begin{theorem}\Label{TT6}
When the message set ${\cal X}$ is the finite field $\FF_q$
and $q$ is a power of $2$,
there exists a imperfectly secure linear code for passive attack.
However, when the message set ${\cal X}$ is the prime field $\FF_2$,
there exists no imperfectly secure linear code for active attack.
\end{theorem}

\subsection{Imperfectly secure non-linear code for passive attack}\Label{S3-1}
First, to show Theorem \ref{TT6},
we give a special example of our code on the prime field $\FF_2$, in which,
the intermediate node performs a non-linear operation as
\begin{align}
Y_3&:= Y_1(Y_2+Y_1)=Y_1(Y_2+1), \Label{Eq1}\\
Y_4&:=(Y_1+1)(Y_2+Y_1)
=(Y_1+1)Y_2. \Label{Eq2}
\end{align}

To send the binary information $M \in \FF_2$,
we prepare the binary uniform scramble random variable $L\in \FF_2$.
We consider the following code.
The encoder $\phi$ is given as 
\begin{align}
Y_1:=L, \quad
Y_2:=M+L. \Label{Eq3}
\end{align}
The decoder $\psi$ is given as $\psi(Y_3,Y_4):=Y_3+Y_4$.
Since $Y_3$ and $Y_4$ are given as follows under this code;
\begin{align}
Y_3= LM, \quad
Y_4=LM+M,
\end{align}
the decoder can recover $M$ whatever the value of $L$.

Now, we consider the leaked information for the passive attack.
As shown in Appendix,
the mutual information and the $l_1$ norm security measure
of these cases are calculated as
\begin{align}
 & I(M;Y_1,Y_3)= I(M;Y_1,Y_4) \nonumber \\
  =& I(M;Y_2,Y_3)=  I(M;Y_2,Y_4)=\frac{1}{2},\Label{E9}\\
 & d_1(M|Y_1,Y_3)= d_1(M|Y_1,Y_4) \nonumber \\
  =& d_1(M|Y_2,Y_3)=  d_1(M|Y_2,Y_4)=\frac{1}{2},\Label{F10}
\end{align}
where the $l_1$ norm security measure $d_1(X|Y)$ is defined as
$d_1(X|Y):=\sum_y \sum_{x}
|\frac{1}{|{\cal X}|}P_Y(y) -P_{XY}(xy)|$
by using the cardinality $|{\cal X}|$ of the set of outcomes of the variable $X$.
In this section, we choose $2$ as the base of the logarithm.
The values in the above relations does not depend on the choice of 
the pair of the eavesdropping edges.
When Eve is allowed to use the above passive attack, \eqref{E9} shows that the code is secure in the sense of (B1).
Therefore, we obtain the first part of Theorem \ref{TT6} when $q=2$.
When $q$ is a power $2^l$ of $2$,
$l$ repetitions of the above non-linear code on the prime field $\FF_2$
give a an imperfectly secure non-linear code on the finite filed $\FF_q$.
Hence, we obtain the first part of Theorem \ref{TT6} even when $q$
is a power of $2$.

Next, we consider the active attack, which is classified into the following three cases.
\begin{description}
\item[(i)]
When $E_E=\{e(1),e(3)\}$,
Eve replaces $Y_1$ by $1$. Then, $I(M;Y_1,Y_3)=1 $ because $Y_3+Y_1+1=M$.

\item[(ii)]
When $E_E=\{e(1),e(4)\}$,
Eve replaces $Y_1$ by $0$. Then, $I(M;Y_1,Y_4)=1 $ because $Y_4+Y_1=M$.

\item[(iii)]
When $E_E=\{e(2),e(3)\}$ or $\{e(2),e(4)\}$,
Eve has no advantageous active attack.
\end{description}
The cases (i) and (ii) show that this code is insecure under the above active attack.
Hence, this example can be regarded as a counterexample of the existing result \cite{HOKC1}[Theorem 1] without linearity.



\begin{remark}
As another encoder, we can consider 
\begin{align}
Y_1:=M+L, \quad
Y_2:=L.
\end{align}
Replacing $M+L$ by $L$, 
the analysis can be reduced to the presented analysis.
When the message $M$ is not leaked to $e(1)$ or $e(2)$ 
and $M$ can be recovered,
the code is essentially the same as our code as follows.

Assume that the information $Y_2$ on $e(2)$ is independent of $M$.
Then, we denote it by $L$.
In order that $Y_1$ is independent and $M$ can be recovered from $Y_1$ and $L$,
$Y_1$ needs to be $M+L$ or $M+L+1$.
\end{remark}

In this model, Eve can completely contaminate the message $M$ as follows.
When Eve takes choice (i), and replaces $Y_3$ by $Y_3+1$,
Bob's decoded message is $M+1$.
Under choice (ii), Eve can totally contaminate the message $M$ in a similar way.

\subsection{Uniqueness of network code given in \eqref{Eq1} and \eqref{Eq2}}\label{S2-5}
The previous subsubsection provided an example where Eve's active attack improves her performance.
To show the second part of Theorem \ref{TT6}, 
we need to show the following lemma.


\begin{lemma}\Label{T6}
Assume that a code $(\phi,\varphi,\psi)$ satisfies the following conditions.
Let $Y_1$ and $Y_2$ be the random variable generated by the encoder 
$\phi$ when $M$ is subject to the uniform distribution.
We assume that the random variables $(Y_3,Y_4):=\varphi (Y_1,Y_2)$
satisfies the following conditions.
\begin{description}
\item[(C1)] The relation $\psi (Y_3,Y_4)=M$ holds.
\item[(C2)] There is no deterministic function $\tilde{\psi}$
from $\FF_2^2$ to $\FF_2$ satisfying one of the following conditions.
\begin{align}
\tilde{\psi}(Y_1,Y_3)=M,\quad
\tilde{\psi}(Y_1,Y_4)=M, \\
\tilde{\psi}(Y_2,Y_3)=M,\quad
\tilde{\psi}(Y_2,Y_4)=M.
\end{align}
\end{description}
Then, there exist functions
$f_1,f_2,f_3,f_4$ on $\FF_2$ such that
$Y_i':=f_i(Y_i)$ is given in \eqref{Eq1}, \eqref{Eq2}, and \eqref{Eq3}
with a scramble random variable $L$ while the variable $L$ might be correlated with $M$.
\end{lemma}

Since 
the number of edges to be attacked
is the same as the transmission rate from Alice to Bob,
no linear code works in this scheme.
Hence, we need to introduce a non-linear coding operation in the intermediate node. 
Lemma \ref{T6} shows that such a non-linear coding operation is limited to 
\eqref{Eq1} and \eqref{Eq2}.
The combination of Lemma \ref{T6} and the discussion in Subsection \ref{S3-1}
implies that there is no code over the one-hop relay network (Fig. \ref{F1})
to guarantee the secrecy for an active attack.
Hence, we obtain the remaining part of Theorem \ref{TT6}.

However, this theorem assumes a deterministic coding operation on the intermediate node.
If we are allowed to use a randomized operation on the intermediate node 
in a similar way to an encoder,
we can construct a code whose secrecy holds even against Eve's active attack
in this situation as follows.
In the first step, we employ the code given in \eqref{Eq3}.
Using another scramble random variable $L'$,
the intermediate node performs the following coding operation:
\begin{align}
Y_3:=Y_1+Y_2+L'=M_L', \quad
Y_4:=L'.
\end{align}
Then, Eve cannot recover the message $M$ from 
$(Y_1,Y_3)$, $(Y_1,Y_4)$, $(Y_2,Y_3)$, nor $(Y_2,Y_4)$.
This example shows that the deterministic condition for $\varphi$ is crucial in Lemma \ref{T6}.

\begin{proofof}{Lemma \ref{T6}}

\noindent{\bf Step (1):}\quad
To satisfy condition (C1),
we need to recover the message $M$ from $(Y_1,Y_2)$ from a deterministic function $f$.
Functions from $\FF_2^2$ to $\FF_2$ are classified as follows.
\begin{align}
&Y_1,~ Y_1+1,Y_2,~ Y_2+1,0,1\Label{D1}\\
&Y_1+Y_2, ~ Y_1+Y_2+1, \Label{D2}\\
&Y_1Y_2,
(Y_1+1)(Y_2+1),
(Y_1+1)Y_2,
Y_1(Y_2+1),
\Label{C3}\\
\begin{split}
&Y_1Y_2+1,
(Y_1+1)(Y_2+1)+1,
(Y_1+1)Y_2+1, \\ 
&Y_1(Y_2+1)+1.
\end{split}
\Label{C4}
\end{align}
The cases in \eqref{D1} are non-secure or
do not satisfy (C1).
The cases in \eqref{D2} are reduced to the case $M=Y_1+Y_2$.
The cases in \eqref{C3} and \eqref{C4} 
are reduced to the case $M=Y_1Y_2$.
Hence, we consider only these two cases.

\noindent{\bf Step (2):}\quad
Now, we consider the case $M=Y_1+Y_2$.
When $Y_3$ or $Y_4$ is given as a non-zero linear function of $Y_1$ and $Y_2$,
we denote the random variable as $Y_*$.
Hence, $Y_1$ or $Y_2$ is linearly independent of $Y_*$.
We denote the linearly independent variable as $Y_{**}$.
When Eve eavesdrops $Y_*$ and $Y_{**}$, 
she can recover $Y_1$ and $Y_2$ and so she can also recover $M$.
To satisfy Condition (C2), we need to avoid such an attack, 
which requires both $Y_3$ and $Y_4$ to be non-linear functions of $(Y_1,Y_2)$.
They are given as two of the functions given in \eqref{C3} and \eqref{C4}.
Since any function in \eqref{C4} is deterministically given from 
a function given in \eqref{C3},
we consider only functions in \eqref{C3}.
Under this constraint,
if and only if $(Y_3,Y_4)$ is given as 
the pair 
$(Y_1Y_2,(Y_1+1)(Y_2+1))$ or $(Y_1(Y_2+1),(Y_1+1)Y_2)$,
we can recover $M=Y_1+Y_2$ from $Y_3$ and $Y_4$.
The latter case is the same as \eqref{Eq1} and \eqref{Eq2}. 
In the former case, we obtain \eqref{Eq1} and \eqref{Eq2}
by replacing $Y_2$ by $Y_2+1$.

\noindent{\bf Step (3):}\quad
Now, we consider the case where $M=Y_1Y_2$.
For the same reason as with Step (2), 
condition (C2) requires 
both $Y_3$ and $Y_4$ to be non-linear functions of 
$(Y_1,Y_2)$.
Thus, we consider only functions in \eqref{C3}.
For secrecy, i.e., to satisfy (C2), we cannot use $Y_1Y_2$.
Hence, we need to choose two from
$(Y_1+1)(Y_2+1),(Y_1+1)Y_2$, and $Y_1(Y_2+1)$.
However, no two of them can recover $M$.
To observe this fact, we consider cases with $(Y_1+1)Y_2 $ and $Y_1(Y_2+1)$.
In these cases, when $(Y_1,Y_2)=(0,0)$ or $(1,1)$,
both values are zero.
That is, we cannot distinguish $(0,0)$ and $(1,1)$.
Hence, we cannot recover $M$ from $(Y_1+1)Y_2 $ and $Y_1(Y_2+1)$, i.e.,
condition (C1) does not hold.
We can show this fact in other pairs in the same way.
Therefore, there is no operation satisfying the required conditions
when $M=Y_1Y_2$.
\end{proofof}

\section{Analysis when characteristic is not 2}\Label{S3-5}
The aim of this section is to prove the statements.

\begin{theorem}\Label{TT7}
When the message set ${\cal X}$ is the quotient ring $\ZZ_d$ with $d\ge 3$, 
there exists an imperfectly secure non-linear code even for active attack.
\end{theorem}

When $q$ is $p^l$,
$l$ repetitions of the above non-linear code on the prime field $\FF_p$
give an imperfectly secure non-linear code on the finite filed $\FF_q$.
Hence, we obtain the following corollary.
\begin{corollary}\Label{TT72}
When the message set ${\cal X}$ is the finite field $\FF_q$ and $q$ is not a power of $2$,
there exists an imperfectly secure non-linear code even for active attack.
\end{corollary}

\subsection{Construction of imperfectly secure code for active attacks}\label{S3-5-1}
To show Theorem \ref{TT7}, we construct a secure network coding against any active attack 
on the one-hop relay network given in Fig. \ref{F1}
when the message and the information at the edges
are given as an element of $\ZZ_d$.
Here, we define our code $(\phi,\varphi,\psi)$
in the same way as in Subsection \ref{S2-5}.
That is, the coding operation $\varphi$ on the intermediate node
has no additional scramble random variable. 

It shows Theorem \ref{TT7}, it is sufficient to construct a code to satisfy the conditions (B1) and (B2) given in Subsection \ref{S4-1}.
Since it is not so easy to check the conditions (B1) and (B2),
we seek equivalent conditions.
For simplicity, 
we employ a scramble variable $L$ taking values in $\ZZ_d$.
Hence, we assume that the encoder $\phi$ in the source node is given as 
a pair of functions $(\phi^{(1)},\phi^{(2)})$
that maps two random variables $(M,L)$
to the two variables $(Y_1,Y_2)$.
That is, the encoder $\phi$ forms a function from $\ZZ_d^2$ to itself.
Now, we fix the function $\phi$ as follows
\begin{align}
Y_1= \phi^{(1)}(M,L):= M+L,\quad 
Y_2= \phi^{(2)}(M,L):=L.\label{con1}
\end{align}

For a coding operation $\varphi$, we define 
the functions $\varphi^{(3)}$ and $\varphi^{(4)}$ as 
\begin{align}
(\varphi^{(3)}(i,j),\varphi^{(4)}(i,j)):=\varphi(i,j).
\end{align}
Then, we regard 
the functions $\varphi^{(3)}$ and $\varphi^{(4)}$ as matrices as follows,
\begin{align}
\varphi^{(3)}_{i,j}:=\varphi^{(3)}(i,j),\quad
\varphi^{(4)}_{i,j}:=\varphi^{(4)}(i,j).
\end{align}
To discuss condition (B2),
we introduce an anti-Latin square.
A matrix $a_{i,j}$ on $\ZZ_d$ is called an {\it anti-Latin square}
when each row and each column have duplicate elements as
\begin{align}
& \left(
\begin{array}{ccc}
1 & 0 & 0 \\
0 & 0 & 2 \\
1 & 2 & 2
\end{array}
\right),\quad
\left(
\begin{array}{ccc}
1 & 0 & 1 \\
1 & 2 & 1 \\
0 & 2 & 0
\end{array}
\right),\Label{Ex1}\\
& \left(
\begin{array}{cccc}
1 & 0 & 3 & 3 \\
0 & 0 & 2 & 3 \\
1 & 1 & 3 & 2 \\
0 & 2& 2 & 1
\end{array}
\right),\quad
\left(
\begin{array}{cccc}
2 & 2 & 1 & 0\\
0 & 3 & 3 & 1 \\
0 & 3 & 3 & 0 \\
1 & 1 & 2 & 2 
\end{array}
\right), \Label{Ex2}
\end{align}
which is the opposite requirement to a Latin square.
Therefore, we have the following lemma.

\begin{lemma}\Label{LOF}
When the encoder $\phi$ satisfies condition \eqref{con1},
conditions (B1) and (B2) are rewritten as
\begin{description}
\item[(B1')] 
For each $a \in \ZZ_d$ and
$m\neq m'\in \ZZ_d$,
the relation $
\Xi_{a,m}(\varphi^{(3)},\varphi^{(4)})
\cap
\Xi_{a,m'}(\varphi^{(3)},\varphi^{(4)})
=\emptyset$ holds,
where $\Xi_{a,m}(\varphi^{(3)},\varphi^{(4)}):=
\varphi^{(4)} (\{(i,i+m)| \varphi^{(3)}_{i,i+m}=a\})$.

\item[(B2')] 
The matrices $\varphi^{(3)}$ and $\varphi^{(4)}$ are anti-Latin squares.
\end{description}
\end{lemma}

\begin{proof}
We have the equivalence between conditions (B1) and (B1')
because (B1') means that the pair of $\varphi^{(3)}(i,j)$ and 
$\varphi^{(4)}(i,j)$ uniquely identifies the difference $m=j-i$.

Next, we show the equivalence between (B2) and (B2').
Assume that Eve eavesdrops and contaminates $Y_1$ and eavesdrops $Y_3$.
Choosing the replaced information $Y_1'$, 
Eve can choose a row of the matrix $\varphi^{(3)}$.
To prevent Eve from recovering $M$ perfectly,
all the rows of the matrix $\varphi^{(3)}$ need to have duplicate elements.
Hence, to satisfy condition (B2),
both matrices $\varphi^{(3)}$ and $\varphi^{(4)}$ need to satisfy this duplication requirement  for all rows and columns.
\end{proof}

Due to Lemma \ref{LOF}, to show Theorem \ref{TT7},
it is sufficient to construct a pair of anti-Latin squares to satisfy conditions (B1') and (B2').
While it is trivial to find anti-Latin squares, they need to satisfy condition (B1') as well.
Condition (B2') forbids a linear operation on the intermediate node in the finite field case.
A pair of anti-Latin squares is called decodable when it satisfies condition (B1').
That is, a decodable pair of anti-Latin squares gives a code on the one-hop relay network given in Fig. \ref{F1} satisfying conditions (B1) and (B2).
Lemma \ref{T6} says that there is no decodable pair of $2 \times 2$ anti-Latin squares.
Fortunately, Eq. \eqref{Ex1} (Eq. \eqref{Ex2})
is a decodable pair of  $3 \times 3$ ($4 \times 4$) anti-Latin squares. 

However, we can systematically construct decodable pairs of anti-Latin squares.
The following are pairs of  anti-Latin squares for $d=3,5,7$:
\begin{align}
&\left(
\begin{array}{ccc}
0 & 1 & 0 \\
1 & 1 & 2 \\
0 & 2 & 2
\end{array}
\right),\quad
\left(
\begin{array}{ccc}
0 & 2 & 2 \\
0 & 1 & 0 \\
1 & 1 & 2 
\end{array}
\right),\Label{Ex3}
\\
&\left(
\begin{array}{ccccc}
0 & 1 & 2 & 0 &0 \\
1 & 1 & 2 & 3 & 1\\
2 & 2 & 2 & 3 & 4 \\
0 & 3 & 3 & 3 & 4 \\
0 & 1 & 4 & 4 & 4 
\end{array}
\right),\quad
\left(
\begin{array}{ccccc}
0 & 3 & 3 & 3 & 4 \\
0 & 1 & 4 & 4 & 4 \\
0 & 1 & 2 & 0 &0 \\
1 & 1 & 2 & 3 & 1\\
2 & 2 & 2 & 3 & 4 
\end{array}
\right),\Label{Ex5}
\\
&\left(
\begin{array}{ccccccc}
0 & 1 & 2 & 3 & 0 & 0 &0 \\
1 & 1 & 2 & 3 & 4 & 1 &1 \\
2 & 2 & 2 & 3 & 4 & 5 &2 \\
3 & 3 & 3 & 3 & 4 & 5 & 6 \\
0 & 4 & 4 & 4 & 4 & 5 & 6 \\
0 & 1 & 5 & 5 & 5 & 5 & 6 \\
0 & 1 & 2 & 6 & 6 & 6 & 6 
\end{array}
\right),\nonumber \\
&\left(
\begin{array}{ccccccc}
0 & 4 & 4 & 4 & 4 & 5 & 6 \\
0 & 1 & 5 & 5 & 5 & 5 & 6 \\
0 & 1 & 2 & 6 & 6 & 6 & 6 \\
0 & 1 & 2 & 3 & 0 & 0 &0 \\
1 & 1 & 2 & 3 & 4 & 1 &1\\
2 & 2 & 2 & 3 & 4 & 5 &2 \\
3 & 3 & 3 & 3 & 4 & 5 & 6 
\end{array}
\right).\Label{Ex7}
\end{align}
These constructions are generalized to the case with a general odd number
$d=2\ell+1$ as follows.
The functions 
$\varphi^{(3)}$ 
and
$\varphi^{(4)}$ 
are defined as
\begin{align}
(\varphi^{(3)})^{-1}(k)
&:=
\left\{
\begin{array}{l}
(k,k-\ell),(k,k-\ell+1),\ldots, \\
(k,k-1),(k,k),(k-1,k),\ldots,\\
(k-\ell+1,k),
(k-\ell,k)
\end{array}
\right\} \\
(\varphi^{(4)})^{-1}(k)
&:=
\left\{
\begin{array}{l}
(k+\ell,k-\ell),(k+\ell,k-\ell+1),\\
\ldots, (k+\ell,k-1),\\
(k+\ell,k),(k+\ell-1,k),\\
\ldots,
(k+1,k),
(k,k)
\end{array}
\right\} .
\end{align}
Then, we have
\begin{align}
\varphi^{(4)}(k,k-\ell)&=k-\ell,\\
\varphi^{(4)}(k,k-\ell+1)&=k-\ell+1,\\
\vdots \nonumber\\
\varphi^{(4)}(k,k-1)&= k-1\\
\varphi^{(4)}(k,k)&=k\\
\varphi^{(4)}(k-1,k)&=k+\ell\\
\vdots \nonumber\\
\varphi^{(4)}(k-\ell+1,k)&=k+2\\
\varphi^{(4)}(k-\ell,k)&=k+1,
\end{align}
which satisfy condition (B1').
Hence, the functions $\varphi^{(3)}$ and $\varphi^{(4)}$ give a pair of  anti-Latin squares. 

Next, we consider the even case with $d \ge 4$.
The following are pairs of  anti-Latin squares for $d=4,6,8$:
\begin{align}
&\left(
\begin{array}{cccc}
0 & 1 & 3 & 3 \\
0 & 1 & 2 & 0 \\
1 & 1 & 2 & 3 \\
0 & 2& 2 & 3
\end{array}
\right),\quad
\left(
\begin{array}{cccc}
0 & 0 & 1 & 0 \\
1 & 1 & 1 & 2 \\
3 & 2 & 2 & 2 \\
3 & 0& 3 & 3
\end{array}
\right), \Label{Ex4}
\\
&\left(
\begin{array}{cccccc}
0 & 1 & 2 & 5 & 5 &5\\
0 & 1 & 2 & 3 & 0&0\\
1 & 1 & 2 & 3 &  4&1\\
2 & 2& 2 & 3 & 4 & 5\\
0 & 3 & 3 &3 & 4 & 5\\
0 & 1  & 4 & 4 & 4 & 5 
\end{array}
\right),\quad
\left(
\begin{array}{cccccc}
1 & 1 & 1 & 2 & 3 & 1\\
2 & 2 & 2 & 2 & 3& 4\\
5 & 3& 3 & 3 & 3 & 4\\
5 & 0 & 4 &4 & 4 & 4\\
5 & 0  & 1 & 5 & 5 & 5 \\
0 & 0 & 1 & 2 & 0 & 0
\end{array}
\right), \Label{Ex6}
\\
&\left(
\begin{array}{cccccccc}
0 & 1 & 2 & 3 & 7 &7 & 7 &7\\
0 & 1 & 2 & 3 & 4&0 & 0& 0\\
1 & 1 & 2 & 3 &  4&5 & 1 & 1\\
2 & 2& 2 & 3 & 4 & 5 & 6 & 2\\
3 & 3 & 3 &3 & 4 & 5 & 6 & 7\\
0 & 4  & 4 & 4 & 4 & 5 & 6&  7 \\
0 & 1  & 5 & 5 & 5 & 5 & 6&  7 \\
0 & 1  & 2 & 6 & 6 & 6 & 6&  7 
\end{array}
\right),\nonumber \\ 
&\left(
\begin{array}{cccccccc}
2 & 2 & 2 & 2 &  3&4 & 5 & 2\\
3 & 3& 3 & 3 & 3 & 4 & 5 & 6\\
7 & 4 & 4 &4 & 4 & 4 & 5 & 6\\
7 & 0  & 5 & 5 & 5 & 5 & 5 & 6 \\
7 & 0  & 1 & 6 & 6 & 6 & 6&  6 \\
7 & 0  & 1 & 2 & 7 & 7 & 7&  7 \\
0 & 0 & 1 & 2 & 3 &0 & 0 &0\\
1 & 1 & 1 & 2 & 3& 4 & 1& 1\\
\end{array}
\right). \Label{Ex8}
\end{align}
These constructions are generalized to the case with a general even number
$d=2\ell$ with $\ell \ge 2$ as follows.
The functions 
$\varphi^{(3)}$ 
and
$\varphi^{(4)}$ 
are defined as
\begin{align}
(\varphi^{(3)})^{-1}(k)
&:=
\left\{
\begin{array}{l}
(k+1,k-\ell+1),\\
(k+1,k-\ell+2),\ldots, \\
(k+1,k-1),(k+1,k),\\
(k,k),(k-1,k),
\ldots,\\
(k-\ell+2,k),
(k-\ell+1,k)
\end{array}
\right\} \\
(\varphi^{(4)})^{-1}(k)
&:=
\left\{
\begin{array}{l}
(k-\ell+1,k-\ell+2),\\
(k-\ell+2,k-\ell+2),
\ldots,\\ 
(k,k-\ell+2),\\
(k+1,k-\ell+2),\\
(k+1,k-\ell+1),
\ldots,\\
(k+1,k-2\ell+4),\\
(k+1,k-2\ell+3)
\end{array}
\right\} .
\end{align}
Then, we have
\begin{align}
\varphi^{(4)}(k+1,k-\ell+1)&=k+\ell,\\
\varphi^{(4)}(k+1,k-\ell+2)&=k+\ell+1,\\
\vdots \nonumber\\
\varphi^{(4)}(k+1,k-1)&= k+2\ell-2\\
\varphi^{(4)}(k+1,k)&=k+\ell-1\\
\varphi^{(4)}(k,k)&=k+\ell-2\\
\varphi^{(4)}(k-1,k)&=k+\ell-3\\
\vdots \nonumber\\
\varphi^{(4)}(k-\ell+2,k)&=k \\
\varphi^{(4)}(k-\ell+1,k)&=k-1,
\end{align}
which satisfy condition (B1').
Hence, the functions $\varphi^{(3)}$ and $\varphi^{(4)}$ give a pair of  anti-Latin squares. 

In summary, since these examples work with $d\ge 3$, 
we have proven Theorem \ref{TT7}, i.e., 
there exists a secure code over the active attacks on the one-hop relay network (Fig. \ref{F1})
when $d\ge 3$.

Furthermore, when $\varphi$ is given by these pairs of anti-Latin squares, 
Bob can decode $L$ as well as $M$
while the code given by \eqref{Ex1} or \eqref{Ex2} cannot.
That is, these systematic constructions work well whenever the encoder
$\phi=(\phi^{(1)},\phi^{(2)})$ is a one-to-one function 
on $\ZZ_d^2$
to satisfy the condition
$
\{i|\exists j, ~\phi^{(1)}{(i,j)}=k\}=
\{i|\exists j, ~\phi^{(2)}{(i,j)}=k\}=
\ZZ_d$ for any $k$.


\subsection{Leaked information of our code for passive attacks}
Next, we discuss the leaked information under the above code under passive attacks
in a way similar to \eqref{E9} and \eqref{F10}.

Since for $i=1,2$ and $j=3,4$, the pair $M,Y_i$ decides $Y_1,Y_2$,
we have $H(Y_j|M Y_i)=0$.
Since $Y_i$ is independent of $M$,
we have
\begin{align}
&I(M;Y_i Y_j)
= H(M)-H(M|Y_iY_j) \nonumber \\
=& H(M|Y_i)-H(M|Y_iY_j)
=I(M; Y_j|Y_i)\nonumber \\ 
=& H(Y_j|Y_i)-H(Y_j|M Y_i)
= H(Y_j|Y_i).
\end{align}
When $d$ is odd,
we have
\begin{align}
&H(Y_j|Y_i)= H(Y_j|Y_i=y_i) \nonumber \\
=& 
\frac{d+1}{2}\cdot
\frac{1}{d}\log \frac{2d}{d+1}
+
\frac{d-1}{2}\cdot
\frac{1}{d}\log d
\end{align}
for any $y_i$ with $i=1,2$ and $j=3,4$.

When $d\ge 4$ and $d$ is even,
we have
\begin{align}
&H(Y_3|Y_2)=H(Y_3|Y_2=y_2)
= H(Y_4|Y_1)\nonumber \\
=& H(Y_4|Y_1=y_1)=
\frac{d+2}{2} \cdot
\frac{1}{d}\log \frac{2d}{d+2}
+
\frac{d-2}{2}\cdot
\frac{1}{d}\log d ,\\
&H(Y_3|Y_1)=H(Y_3|Y_1=y_1)= H(Y_4|Y_2)
\nonumber \\
=& H(Y_4|Y_2=y_2)
= 
\frac{1}{2}\log 2
+
\frac{1}{2}\log d 
\end{align}
for any $y_1,y_2$.
In summary, when $d$ is large, we have
\begin{align}
I(M;Y_i Y_j) = \frac{1}{2}\log d +\frac{1}{2}\log 2
+ O(\frac{\log d}{d}).\Label{F29-6}
\end{align}

\subsection{Lower bound of leaked information for passive attacks}\Label{S3-6}
Next, to show the optimality of the code defined in Subsection \ref{S3-5-1}, 
we show that the amount in \eqref{F29-6} is close to the minimum leaked information
under a certain condition when $d$ is large.
To derive a lower bound, we consider the following conditions for our code.

\begin{description}

\item[(D1)]
The coding operation on the intermediate node is deterministic.

\item[(D2)]
Alice can use a scramble random variable $L$.
\end{description}
Since our encoder is given as a stochastic map $\phi$
from $\ZZ_2$ to $\ZZ_2^2$ in Section \ref{S4-1},
condition (D2) is a more restrictive condition for our encoder.
Then, we have the following theorem.

\begin{lemma}\label{TD1}
Any network code satisfies 
the inequality
\begin{align}
I(M; Y_i Y_3)+I(M; Y_i Y_4) \ge 2 H(M) - \log d
\end{align}
for $i=1,2$.
\end{lemma}

\begin{proof}
Since $M$ is decoded by $Y_3 Y_4$,
\begin{align*}
& 
H(M| Y_i Y_3)
\le H(Y_3Y_4 | Y_i Y_3)
= H(Y_4 | Y_i Y_3)
 \le H(Y_4| Y_i ).
\end{align*}
Similarly, we have $H(M| Y_i Y_4) \le H(Y_3| Y_i Y_4)$ by replacing $Y_3$ and $Y_4$.
Let $i'$ be the integer $1$ or $2$ that is different from $i$.
Combining them, we have
\begin{align*}
& H(M| Y_i Y_3)+H(M| Y_i Y_4)
 \le H(Y_4| Y_i )+H(Y_3| Y_i Y_4) \\
 =& H(Y_3 Y_4| Y_i ) 
 \stackrel{(a)}{\le}
 H(Y_i Y_{i'}| Y_i ) \\
=& H(Y_{i'}| Y_i ) \le \log d,
\end{align*}
where $(a)$ follows from the fact that 
$Y_3 Y_4$ is decided by $Y_1 Y_2$.
Thus, we obtain
\begin{align*}
&I(M; Y_i Y_3)+I(M; Y_i Y_4) \\
=& 2 H(M) - (H(M| Y_i Y_3)+H(M| Y_i Y_4)) \\
\ge & 2 H(M) - \log d.
\end{align*}
\end{proof}

Lemma \ref{TD1} shows that
\begin{align}
\max_{i,j}I(M; Y_i Y_j) \ge H(M)- \frac{1}{2}\log d,\Label{NHT}
\end{align}
where the maximum is chosen from $i=1,2$ and $j=3,4$.
That is, to realize $\max_{i,j}I(M; Y_i Y_j)=0$, 
the message $M$ needs to satisfy 
\begin{align}
H(M)\le \frac{1}{2}\log d.
\end{align}
When $M$ is the uniform random variable, \eqref{NHT} can be rewritten as 
\begin{align}
\max_{i,j}I(M; Y_i Y_j) \ge \frac{1}{2}\log d.
\end{align}
This lower bound is almost equal to the RHS of \eqref{F29-6} when $d$ is large.

\section{Conclusion}\Label{SCon}

\begin{table}[htpb]
  \caption{Summary for security analysis over the one-hop relay network}
\Label{non-linear}
\begin{center}
  \begin{tabular}{|l|c|c|} 
\hline
Code & passive attack & active attack\\
\hline
\hline
linear code when ${\cal X}$ is $\FF_q$ & insecure & insecure \\
\hline
non-linear code with $\FF_2$ & imperfectly secure & insecure \\
\hline
non-linear code with $\FF_q$ 
& \multirow{2}{*}{imperfectly secure} & \multirow{2}{*}{imperfectly secure}  \\
and $q\neq 2^l$ &&\\
\hline
  \end{tabular}
\end{center}
\end{table}

In this paper, we have discussed how sequential error injection affects the information leakage over the one-hop relay network.
For this aim, we have studied secure network coding for active attacks 
over the one-hop relay network (Fig. \ref{F1}) including non-linear operations on the intermediate node.
The obtained results are summarized in Table \ref{non-linear}.

First, we show that no linear code realizes the security (even imperfect security)
when Eve eavesdrops one channel in both layers, the intermediate node can make only a deterministic operation, and 
the message set ${\cal X}$ is a finite field $\FF_q$.
Hence, to meet the security, 
the coding operation on the intermediate node needs to be non-linear.
When the message set ${\cal X}$ is the finite field $\FF_2$,
we propose a code to satisfy the imperfect security under the passive attack by Eve.
However, it does not satisfy the imperfect security under the active attack.
As shown in Section \ref{S2-5}, any code to satisfy the imperfect security is limited to 
a code equivalent to this non-linear code.
Hence, unfortunately, no code on the one-hop relay network is imperfectly secure under active attacks
in the binary case.

To realize the imperfect security even under active attacks,
we have studied the case when 
the message set ${\cal X}$ is the finite field $\FF_q$ with $q \neq 2^l$, which includes the ternary case. 
This can be obtained as a simple extension of the case with the prime field $\FF_p$.
To characterize the desired security in this case, 
we have introduced a new concept an ``anti-Latin square'', which is an opposite concept to a Latin square.
That is, such a secure code can be given as a decodable pair of anti-Latin squares while the concept of ``decodable'' is also introduced in Section \ref{S3-5}.
We have also shown the existence of a decodable pair of 
$p \times p$ anti-Latin squares when $p \ge 3$.
This fact shows that 
there exists an imperfectly secure code over active attacks except for the binary case.
The obtained analysis of this part
still holds when the message set ${\cal X}$ is $\ZZ_d$.


In fact, it is still remained an open problem whether there exists an imperfectly secure 
code over active attack on the one-hop relay network
when the message set ${\cal X}$ is $\FF_{2^l}$ with $l \ge 2$.
This is an interesting future study.
Also, the obtained analysis depends on the property of 
the single transmission scheme.
As shown in the paper \cite{CH}, under the setting of transmission of multiple letters, 
there exists a secure code to satisfy the vector linearity.
Since our analysis is limited to the one-hop relay network,
it is a challenging future study to investigate 
a similar security analysis over a multiple-layered network.

Due to our analysis, a secure code on our model is limited to 
non-linear code to satisfy the conditions (B1') and (B2').
Clearly, as mentioned in the end of Subsection \ref{S3-1},
this type of codes do not satisfy the robustness when 
Eve contaminates the information on the attacked edges.
Hence, to realize the robustness as well as the secrecy,
we need to increase the number of edges between nodes.
Since the robustness is another important issue in secure network coding,
it is another interesting future study to the secrecy and the robustness jointly over a modified network model.

\section*{Acknowledgments}
MH and NC are very grateful to 
Prof. Masaki Owari, Dr. Go Kato, and 
Dr. Wangmei Guo
for helpful discussions and comments.

\appendix

First, we calculate $I(M;Y_1,Y_3)$ and $d_1(M|Y_1,Y_3)$.
We find that
\begin{align*}
P_{Y_1,Y_3|M}(0,0|0)=
P_{Y_1,Y_3|M}(1,0|0) =\frac{1}{2},\\
P_{Y_1,Y_3|M}(0,0|1)=
P_{Y_1,Y_3|M}(1,1|1)=\frac{1}{2},
\end{align*}
where the remaining conditional probabilities are zero.
Hence, 
\begin{align*}
H(Y_1,Y_3|M)=1, ~
H(Y_1,Y_3)=\frac{1}{2}\log 2 +\frac{1}{2}\log 4 =\frac{3}{2},
\end{align*}
which implies $I(M;Y_1,Y_3)=\frac{1}{2}$.

Since
\begin{align*}
P_{Y_1,Y_3}(0,0)=\frac{1}{2},~
P_{Y_1,Y_3}(1,0)=
P_{Y_1,Y_3}(1,1)=\frac{1}{4}, 
\end{align*}
we have
\begin{align*}
P_{M|Y_1,Y_3}(0|0,0)&=
P_{M|Y_1,Y_3}(1|0,0)=\frac{1}{2},\\
P_{M|Y_1,Y_3}(0|1,0)&=
P_{M|Y_1,Y_3}(1|1,1)=1,
\end{align*}
where the remaining conditional probabilities are zero.
Therefore,
\begin{align*}
&d_1(M|Y_1,Y_3) \\
= &
\Big|\frac{1}{2}P(0,0)-\frac{1}{2}P(0|0,0)\Big|
+\Big|\frac{1}{2}P(1,0)-\frac{1}{2}P(0|1,0)\Big| \nonumber \\
&+\Big|\frac{1}{2}P(1,1)-\frac{1}{2}P(0|1,1)\Big| 
+\Big|\frac{1}{2}P(0,0)-\frac{1}{2}P(1|0,0)\Big|\nonumber \\
&
+\Big|\frac{1}{2}P(1,0)-\frac{1}{2}P(1|1,0)\Big|
+\Big|\frac{1}{2}P(1,1)-\frac{1}{2}P(1|1,1)\Big|\nonumber \\
=&0+\frac{1}{8}+\frac{1}{8}
+0+\frac{1}{8}+\frac{1}{8}
=
\frac{1}{2}.
\end{align*}
Replacing $M$ and $L$ by $M+1$ and $L+1$, respectively,
we can calculate $I(M;Y_1,Y_4)$ and $d_1(M|Y_1,Y_4)$
in the same way.

Next, we consider 
$I(M;Y_2,Y_3)$ and $d_1(M|Y_2,Y_3)$.
We find that
\begin{align*}
P_{Y_2,Y_3|M}(0,0|0)=
P_{Y_2,Y_3|M}(1,0|0) =\frac{1}{2},\\
P_{Y_2,Y_3|M}(0,1|1)=
P_{Y_2,Y_3|M}(1,0|1)=\frac{1}{2},
\end{align*}
where the remaining conditional probabilities are zero.
Hence, replacing $(0,0)$ and $(1,1)$ by $(1,0)$ and $(0,1)$, respectively, in the above
derivation, we can show 
$I(M;Y_2,Y_3)=\frac{1}{2}$
and $d_1(M|Y_2,Y_3)=\frac{1}{2}$.
Finally, replacing $M$ and $L$ by $M+1$ and $L+1$, respectively,
we can calculate $I(M;Y_2,Y_4)$ and $d_1(M|Y_2,Y_4)$ in the same way.

\end{document}